\documentclass[aps,prl,showpacs,final,twocolumn]{revtex4}

\usepackage{amsfonts}
\usepackage{amssymb}
\usepackage{amsmath}
\usepackage{amsthm}
\usepackage{mathrsfs}
\usepackage{bm}
\usepackage{graphicx}
\usepackage{multirow}
\usepackage{enumitem}
\theoremstyle{definition} \newtheorem{thm}{Theorem}
\theoremstyle{definition} \newtheorem{cor}{Corollary}

\begin{document}
\title{General method for finding ground state manifold of classical Heisenberg model}

\author{Zhaoxi Xiong$^{1,}$}
\email{xiong@mit.edu}
\author{Xiao-Gang Wen$^{1,2}$}
\affiliation{$^1$Department of Physics, Massachusetts Institute of Technology, Cambridge, Massachusetts 02139, USA\\
$^2$Perimeter Institute for Theoretical Physics, Waterloo, Ontario, N2L 2Y5 Canada}
\date{\today}

\begin{abstract}

We investigate classical Heisenberg models with the translation symmetries of infinite crystals. We prove a spiral theorem, which states that under certain conditions there must exist spiral ground states, and propose a natural classification of all manageable models based on some ``spectral properties," which are directly related to their ground state manifolds. We demonstrate how the ground state manifold can be calculated analytically for all spectra with finite number of minima and some with extensive minima, and algorithmically for the others. We also extend the method to particular anisotropic interactions.

\end{abstract}


\maketitle

Classical spin orders are the starting point of nearly every quantum mechanical treatment for the same subject. In the twentieth century, quantum spin wave theories were developed to analyze the low energy excitations of systems with ferromagnetic (F) and antiferromagnetic (AF) classical ground states \cite{quantumSpinWave, HolsteinPrimakoff}. By introducing local frames, one can extend such quantum fluctuation analysis to systems with non-collinear and non-coplanar classical spin orders (see e.g. \cite{localFrame}).

This ``quantum-classical" approach is suitable for large-$S$ systems such as rare earth materials, which have attracted special interests in the past decade due to their central role in quantum anomalous Hall effect \cite{QAHEmechanism}. In this respect, non-coplanar classical spin orders are especially favored. In the opposite small-$S$ extreme, one either extrapolates the results of large-$S$ expansions, or in the case of spin-$\frac 12$'s takes classical spins as a mean-field approximation.

It is therefore of fundamental importance to understand classical spin orders. Among all existing methods, the most intuitive ones are probably what might be called the pairwise minimization method (see e.g. \cite{pairwiseMinimization}), and its generalized version the cluster method \cite{clusterMinimization}, where the energy is minimized locally, and the global compatibility is essentially left to chance. A different approach is to use weak constraints while minimizing the energy, and to check whether the strong constraints are met afterwards. This is the idea of the Luttinger-Tisza (LT) method \cite{LT} and the generalized Luttinger-Tisza method \cite{GLT} (see also \cite{reviewLTs} for a review). Additionally, there is a so called classical spin wave method (see e.g. \cite{squareJ1J2F,honeycombJ1J2J3}), in which the energy is minimized within a spiral or helical ansatz (see Fig.~\ref{fig:schematic}). The ansatz is partially justified by a spiral theorem proved using the LT methods \cite{reviewLTs}. Superposition of spin waves is valid only in special cases.

\begin{figure}[!b]
\centering
\includegraphics[width=2.1in]{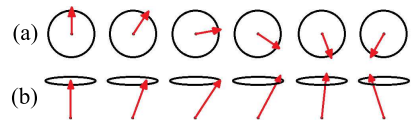}
\caption{Schematic (a) spiral and (b) helical states.}
\label{fig:schematic}
\end{figure}

We are concerned with a classical Heisenberg Hamiltonian with the following properties
(will be referred to as the basic assumptions): ($\rm{i}$) it has the translation symmetry of an infinite crystal, and ($\rm{ii}$) all spins are real, unit 3-vectors. This Hamiltonian is rather general: we made no assumption about the crystal dimensionality, the Bravais lattice, the basis, or the pattern of interaction. It is also obvious that the difference in spin length can be absorbed into the coupling constants. We shall strive to find the {\it entire exact} ground state manifold (GSM) of the Hamiltonian.

While in both the cluster method and the LT methods some of the constraints are first dropped and then restored, we shall take a approach in which full information is retained all the time. This turns out to be more straightforward conceptually, and practically it enables the determination of the entire GSM, which is crucial in numerous contexts. The formalism to be developed will suggest a natural classification of all manageable models (i.e. those in the realm of the spiral theorem) based on the correspondence between some simple spectral properties (to be explained) and the GSMs. Dictionaries from the spectral properties to GSMs could hence be compiled, so that understanding a classical Heisenberg model would become utterly trivial.

The spiral theorem to be proved is the following: for any classical Heisenberg model that satisfies the basic assumptions, if it has one spin per unit cell, or if it has multiple spins per unit cell and has some additional properties, then there is a spiral state in the GSM. The one-spin part has been proved previously in the LT framework \cite{reviewLTs}, but to our knowledge no formal statement of the multi-spin part has ever existed.

\emph{One-spin case.}---Consider a classical Heisenberg system with a large number $N$ of single-spin unit cells. The Hamiltonian can be written as
\begin{equation}
H = \frac12 \sum_{\vec n} \sum_{\vec \eta} J_{\vec \eta} \vec S_{\vec n} \cdot \vec S_{\vec n + \vec \eta} = \frac12 \sum_{\vec n} \sum_{\vec \eta} J_{\vec \eta} \vec S_{\vec n}^\dag \vec S_{\vec n + \vec \eta},
\label{1_H}
\end{equation}
where $J_{\vec \eta}$ are the $\vec \eta$-dependent coupling constants, and real, unit 3-vectors $\vec S_{\vec n}$ are the spins at Bravais lattice points $\vec n$. We shall keep to the conventions of positive sign and factor $\frac12$. Passing to the Fourier space, we rewrite the energy per spin as
\begin{equation}
\epsilon = \sum_{\vec k} \frac 1N \vec S_{\vec k}^\dag \epsilon_{\vec k} \vec S_{\vec k},
\label{1_epsilonDiag}
\end{equation}
where
\begin{eqnarray}
\epsilon_{\vec k} &\equiv& \frac12 \sum_{\vec \eta} J_{\vec \eta} e^{i \vec k \cdot \vec \eta}, \label{1_epsilonk} \\
\vec S_{\vec k} &\equiv& \frac{1}{\sqrt N} \sum_{\vec n} \vec S_{\vec n} e^{- i \vec k \cdot \vec n}.
\label{1_Sk}
\end{eqnarray}
Because Eq.~(\ref{1_H}) formally resembles a hopping Hamiltonian -- in that case each component of $\vec S_{\vec n}$ would have to be interpreted as an operator -- it is appropriate to call $\epsilon_{\vec k}$ the spectrum or the band structure of $H$, and to perceive as the $\mathcal T$-reversal symmetry the fact that $J_j$ are real. It follows from the $\mathcal T$-reversal symmetry that $\epsilon_{-\vec k} = \epsilon_{\vec k}$.

To find the GSM, we need to know exactly what $\{ \vec S_{\vec k} \}$ are legitimate. A configuration $\{ \vec S_{\vec n} \}$ is legitimate if and only if it is contained in
\begin{equation}
{\mathcal D} = \big\{ \{ \vec S_{\vec n} \} | \vec S_{\vec n} \in {\mathbb R}^3,\ \vec S_{\vec n}^\dag \vec S_{\vec n} = 1 \big\},
\label{1_domain}
\end{equation}
It follows that the set of legitimate $\{ \vec S_{\vec k} \}$ is
\begin{eqnarray}
{\mathcal I} = \big\{ \{ \vec S_{\vec k} \} | \vec S_{\vec k} \in {\mathbb C}^3, \vec S_{-\vec k} = \vec S_{\vec k}^*, \sum_{\vec k} \vec S_{\vec k}^\dag \vec S_{\vec k + \vec p} = N \delta_{\vec 0 \vec p}, \forall \vec p \big\},
\label{1_image}
\end{eqnarray}
where the summation over $\vec k$ is taken in the Brillouin zone (BZ), and reciprocal translation vectors are implicitly added to wave vectors that lie outside the BZ. $\mathcal D$ and $\mathcal I$ are the domain and the image of the injective map Eq.~(\ref{1_Sk}), respectively. We are ready to prove the following theorem.

\begin{thm}[Spiral theorem (one-spin case)]
Every classical Heisenberg Hamiltonian with one spin per unit cell that satisfies the basic assumptions has a spiral state in its GSM.
\label{theorem:1spiral}
\end{thm}
\begin{proof}
It follows from the last condition in Eq.~(\ref{1_image}) that $\sum_{\vec k} \frac1N \vec S_{\vec k}^\dag \vec S_{\vec k} = 1$, so if $\vec S_{\vec k} \neq \vec 0$ only at some global minima in the spectrum, then the resulting state must be a rigorous ground state. Suppose $\vec k_1$ is a global minimum. Since $\epsilon_{-\vec k} = \epsilon_{\vec k}$, so is $-\vec k_1$. One can verify that the specification $\vec S_{\pm \vec k_1} = \frac{\sqrt N}{2} (1,\pm i,0)^T$ for inequivalent $\pm \vec k_1$, or $\sqrt N (1,0,0)^T$ for equivalent $\pm \vec k_1$, and zero elsewhere, is allowable according to $\mathcal I$. This corresponds to the configuration $\vec S_{\vec n} = (\cos (\vec k_1 \cdot \vec n), - \sin(\vec k_1 \cdot \vec n), 0)^T$, or $(e^{i \vec k_1 \cdot \vec n}, 0, 0)^T$, respectively, which is a spiral state.
\end{proof}

The absolutely lowest points in the spectrum (out of all bands, in the multi-spin case to come) will be called the ``spectral minima." Theorem \ref{theorem:1spiral} implies that the entire GSM can be found by discarding all $\vec S_{\vec k}$ but those at the spectral minima, and plugging them into Eq.~(\ref{1_image}). It is also clear that there is a direct correspondence between the GSM and the spectral minimum distribution, which can hence be used to classify all one-spin models.

We shall now classify two-dimensional (D) systems and present their GSMs. Given any Bravais lattice, we can align the axes accordingly so that $n_x,\ n_y$ are integers and that the BZ is $(-\pi, \pi] \times (-\pi, \pi]$. Let us restrict ourselves for now to the case of finite number of spectral minima. Due to the periodicity of $\vec k$, some special distributions must be treated separately. Examples of them are given in Fig.~\ref{fig:2DfiniteSpecial}, which is far from complete, of course; yet more complicated distributions appear less often. We summarize the GSMs for these special distributions along with the GSMs for some generic distributions in Table~\ref{table:2Dfinite} \cite{supp}. The six entries in Table~\ref{table:2Dfinite} represent systems with spiral ground states, F/AF ground states, frustrated ground states, spiral ground states, frustrated ground states, and (possibly alternating) helical ground states, respectively. With tables like this established, it would be really easy to read off the GSM of a model for each value of the parameters, and thereby sketch out its ground state phase diagram.

\begin{figure}[!b]
\centering
\includegraphics[width=3.4in]{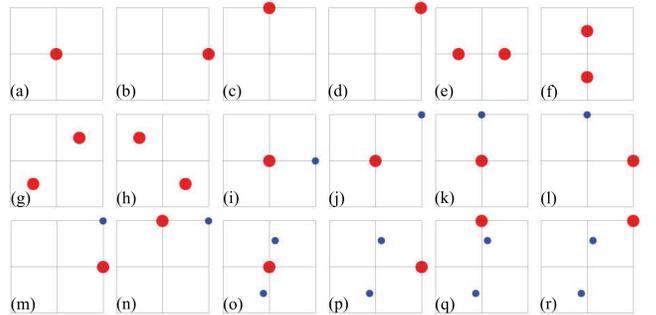}
\caption{Special spectral minimum distributions for 2D systems with one spin per unit cell. The boxes are the BZ $(-\pi, \pi] \times (-\pi, \pi]$ (see text). The dots represent the spectral minima, with the larger and the smaller ones named $\pm \vec k_1$ and $\pm \vec k_2$ respectively. The list exhausts all special distributions with one pair of minima. In (o) - (r), $\pm \vec k_2$ are generically located.}
\label{fig:2DfiniteSpecial}
\end{figure}

\begin{table}[!t]
\caption{The GSMs for spectral minimum distributions of 2D systems with one spin per unit cell. Parameters $\gamma \in \mathbb R,\ \rho \in SO_3$ are arbitrary.}
\begin{tabular}{cc}
\hline
\hline
Spectral minima & GSM\\
\hline
Generic $\pm \vec k_1$ & $\vec S_{\vec n} = \rho \begin{pmatrix} \cos(\vec k_1 \cdot \vec n) \\ -\sin(\vec k_1 \cdot \vec n) \\ 0 \end{pmatrix}$ \\
Figs. \ref{fig:2DfiniteSpecial}(a)-\ref{fig:2DfiniteSpecial}(d) & $\vec S_{\vec n} = \rho \begin{pmatrix} e^{i \vec k_1 \cdot \vec n} \\ 0 \\ 0 \end{pmatrix}$ \\
Figs. \ref{fig:2DfiniteSpecial}(e)-\ref{fig:2DfiniteSpecial}(h) & $\vec S_{\vec n} = \rho \begin{pmatrix} \sqrt2 \cos(\vec k_1 \cdot \vec n + \frac \pi 4 ) \cos \gamma \\ \sqrt2 \sin(\vec k_1 \cdot \vec n + \frac \pi 4 ) \sin \gamma \\ 0 \end{pmatrix}$ \\
Generic $\pm \vec k_1,\ \pm \vec k_2$ & $\vec S_{\vec n} = \rho \begin{pmatrix} \cos(\vec k_j \cdot \vec n) \\ -\sin(\vec k_j \cdot \vec n) \\ 0 \end{pmatrix},\ j=1\mbox{ or }2$\\
Figs. \ref{fig:2DfiniteSpecial}(i)-\ref{fig:2DfiniteSpecial}(n) & $\vec S_{\vec n} = \rho \begin{pmatrix} e^{i \vec k_1 \cdot \vec n} \cos \gamma \\ e^{i \vec k_2 \cdot \vec n} \sin \gamma \\ 0 \end{pmatrix}$ \\
Figs. \ref{fig:2DfiniteSpecial}(o)-\ref{fig:2DfiniteSpecial}(r) & $\vec S_{\vec n} = \rho \begin{pmatrix} \cos(\vec k_2 \cdot \vec n) \sin \gamma \\ -\sin(\vec k_2 \cdot \vec n) \sin \gamma \\ e^{i \vec k_1 \cdot \vec n} \cos \gamma \end{pmatrix}$ \\
\hline
\hline
\end{tabular}
\label{table:2Dfinite}
\end{table}

Let us apply our method to two classic problems. The first is the square lattice $J_1 - J_2$ model, where $J_1$ and $J_2$ are the nearest-neighbor (NN) and the next nearest-neighbor coupling constants, respectively. Its spectrum is $\epsilon_{\vec k} = J_1 (\cos k_x + \cos k_y) + 2 J_2 \cos k_x \cos k_y$. In regimes $J_1 < 0,\ J_2 < 0.5|J_1|$; $J_1>0,\ J_2<0.5 J_1$; and $J_2>0.5 |J_1|$, which fall under Figs.~\ref{fig:2DfiniteSpecial}(a)(d)(l), the system is in F phase, AF phase, and frustrated two-sublattice AF phase, respectively (see e.g. \cite{squareJ1J2AF,squareJ1J2F}). The second is the triangular lattice $J_1 - J_2$ model, where $J_1, ~ J_2$ are again the coupling constants for NNs and next NNs, respectively. The spectrum is $\epsilon_{\vec k} = J_1 [\cos k_x + \cos k_y + \cos(k_x-k_y)] + J_2 [\cos (k_x+k_y)+ \cos (2k_x-k_y) + \cos(2k_y-k_x)]$ (for a particular alignment of the axes). In the regime $J_1 <0,\ J_2 < \frac13 |J_1|$, the minimum distribution is given by Fig.~\ref{fig:2DfiniteSpecial}(a), so we have ferromagnetic ground states. For $J_1 > 0,\ J_2 < \frac18 J_1$, a single pair of minima appears at $\pm (\frac23 \pi, -\frac23 \pi)$, resulting in the famous $\frac{2\pi}{3}$-phase (see e.g. \cite{triangularJ1}). The remaining cases contain three pairs of minima, which are not covered in Table \ref{table:2Dfinite} but are solvable.

\emph{Multi-spin case.}---Consider the following classical Heisenberg Hamiltonian of a system with $N$ unit cells and $m$ sublattices:
\begin{eqnarray}
H &=& \frac12 \sum_{\vec n} \sum_{\vec \eta} \sum_{a,b=1}^m J_{\vec \eta ab} \vec S_{\vec n a} \cdot \vec S_{\vec n + \vec \eta b} \nonumber\\
&=& \frac12 \sum_{\vec n} \sum_{\vec \eta} \sum_{a,b=1}^m J_{\vec \eta ab} \vec S_{\vec n a}^\dag \vec S_{\vec n + \vec \eta b},
\label{m_H}
\end{eqnarray}
where $a,\ b = 1,2,\ldots,m$ are the sublattice indices. In the momentum space, $H$, or more conveniently the energy per spin $\epsilon$, is partially diagonal, reading
\begin{eqnarray}
\epsilon = \sum_{\vec k, a,b} \frac{1}{Nm} \vec S_{\vec k a}^\dag h(\vec k)_{ab} \vec S_{\vec k b}, \label{m_epsilonPartialDiag}
\end{eqnarray}
with
\begin{eqnarray}
\vec S_{\vec k a} &\equiv& \frac{1}{\sqrt N} \sum_{\vec n} \vec S_{\vec n a} e^{-i \vec k \cdot \vec n}, \label{m_Ska} \\
h(\vec k)_{ab} &\equiv& \frac 12 \sum_{\eta} J_{\vec \eta ab} e^{i \vec k \cdot \vec \eta}. \label{m_hk}
\end{eqnarray}
We can further express $h(\vec k)$ in terms of its eigenvalues $\epsilon_{\vec k \alpha}$ and eigenvectors $\xi_{\vec k \alpha}$, where $\alpha = 1,~ \ldots,~ m$, yielding
\begin{eqnarray}
\epsilon = \sum_{\vec k, \alpha} \frac{1}{Nm} \vec S_{\vec k \alpha}^\dag \epsilon_{\vec k \alpha} \vec S_{\vec k \alpha}, \label{m_epsilonDiag}
\end{eqnarray}
with
\begin{equation}
\vec S_{\vec k \alpha} \equiv \sum_{a} \vec S_{\vec k a} \xi_{\vec k \alpha, a}^*. \label{m_Skalpha}
\end{equation}
Multiple bands emerge in this multi-spin scenario. The $\mathcal T$-reversal symmetry implies $h(-\vec k) = h(\vec k)^*$, and consequently, the eigenvalues satisfy $\epsilon_{-\vec k \alpha} = \epsilon_{\vec k \alpha}$, and the eigenvectors can be chosen so that $\xi_{-\vec k \alpha} = \xi_{\vec k \alpha}^*$. Then the image of the composite map Eqs.~(\ref{m_Ska})(\ref{m_Skalpha}) can be written as
\begin{eqnarray}
{\mathcal I} &=&
\big\{
\{ \vec S_{\vec k \alpha} \} | \vec S_{\vec k \alpha} \in {\mathbb C}^3,\ \vec S_{-\vec k \alpha} = \vec S_{\vec k \alpha}^*, \nonumber\\
&& \sum_{\vec k,\alpha,\alpha'} \xi_{\vec k \alpha', a}^* \vec S_{\vec k \alpha'}^\dag \vec S_{\vec k + \vec p \alpha} \xi_{\vec k + \vec p \alpha, a} = N \delta_{\vec 0 \vec p}, \forall a, \vec p
\big\}.
\label{m_image}
\end{eqnarray}
We now present the multi-spin version of the spiral theorem, which involves an additional requirement.

\begin{thm}[Spiral theorem (multi-spin case)]
For any classical Heisenberg Hamiltonian with $m$ spins per unit cell that satisfies the basic assumptions, if there is a spectral minimum $(\vec k_1, \alpha_1)$ (out of all $\vec k, \alpha$)
where $\xi_{\vec k_1 \alpha_1}$ can be chosen to take the form $\xi_{\vec k_1 \alpha_1} = \frac{1}{\sqrt m} ( e^{- i t_1}, \ldots, e^{- i t_m} )^T$, then there exists a spiral state in the GSM.
\label{theorem:mspiral}
\end{thm}
\begin{proof}
Tracing the last condition in Eq.~(\ref{m_image}) over $a$, one finds $\sum_{\vec k, \alpha} \frac{1}{Nm} \vec S_{\vec k \alpha}^\dag \vec S_{\vec k \alpha} = 1$, so if $\vec S_{\vec k\alpha} \neq \vec 0$ only at spectral minima, then the resulting state must be a rigorous ground state. Suppose $(\pm \vec k_1, \alpha_1)$ are a pair of spectral minima where $\xi_{\pm\vec k_1 \alpha}$ take the required form. Then one can verify that the prescription such that $\vec S_{\pm \vec k_1, \alpha_1} = \frac{\sqrt{Nm}}{2} ( 1 , \pm i , 0 )^T$ for inequivalent $\pm k_1$, or $\sqrt{Nm} ( 1 , 0 , 0 )^T$ for equivalent $\pm k_1$, and zero elsewhere, is allowable according to $\mathcal I$. This represents a spiral state, $\vec S_{\vec n a} = \big( \cos(\vec k_1 \cdot \vec n - t_a), -\sin(\vec k_1 \cdot \vec n - t_a), 0 \big)^T$ or $\big( e^{i(\vec k_1 \cdot \vec n - t_a)}, 0, 0 \big)^T$, respectively.
\end{proof}

\begin{cor}
For a system with $m$ spins per unit cell, if for each $\vec k$, up to a gauge transformation (i.e. multiplying basis by phases), $h(\vec k)_{ab}$ only depends on $(a-b)\mbox{ mod }m$, then there exists a spiral ground state.
\label{corollary:m}
\end{cor}

\begin{cor}
For a system with two spins per unit cell, if $h(\vec k)$ contain no $\sigma_z$ for all $\vec k$, or equivalently if $H$ is invariant under an inversion that exchanges the two sublattices, then there exists a spiral ground state.
\label{corollary:2}
\end{cor}

\begin{cor}
For a system with three spins per unit cell, if $h(\vec k)_{11} = h(\vec k)_{22} = h(\vec k)_{33}$ and $|h(\vec k)_{12}| = |h(\vec k)_{23}| = |h(\vec k)_{13}|$ for all $\vec k$, then there exists a spiral ground state.
\label{corollary:3}
\end{cor}

That $h(\vec k)_{ab}$ only depends on $(a-b) \mbox{ mod } m$ means that the sublattice degree of freedom has become an extra finite periodic dimension (think of nanotubes), and yet Corollaries~\ref{corollary:m}-\ref{corollary:3} allow this to be achieved after gauge transformations. The condition in Theorem~\ref{theorem:mspiral} is even weaker, as it only stipulates that $\xi_{\pm \vec k_1 \alpha_1}$ should take a particular form.

\begin{figure}[!t]
\centering
\includegraphics[width=2.8in]{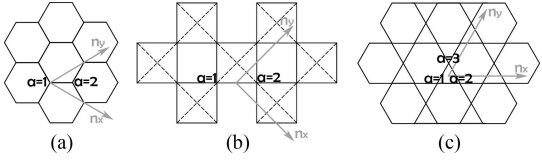}
\caption{Classic 2D models with multiple spins per unit cell. (a) The honeycomb $J_1 - J_2$ model, where $J_1, ~ J_2$ are coupling constants for NNs and next NNs, respectively. (b) The checkerboard $J_1 - J_2$ model, where $J_1, ~ J_2$ are coupling constants for solid and dashed edges, respectively. (c) The Kagome NN model with coupling constant $J_1$.}
\label{fig:2DclassicMulti}
\end{figure}
When the requirement in Theorem~\ref{theorem:mspiral} is met, the GSM is determined by the spectral minimum distribution and the $\xi_{\vec k \alpha}$ thereof. We omit the similar classification here and simply point out some tractable classic models. First, Corollary~\ref{corollary:2} is applicable to the honeycomb $J_1 - J_2$ model (Fig.~\ref{fig:2DclassicMulti}(a)) \cite{localFrame,honeycombJ1J2J3}. Second, the $h(\vec k)$ of the checkerboard $J_1 - J_2$ model (Fig.~\ref{fig:2DclassicMulti}(b)) has no $\sigma_z$ component along $|k_x| = |k_y|$, which happens to include some of the spectral minima in all cases \cite{checkerboardJ1J2AF}. Finally, for the Kagome NN model (Fig.~\ref{fig:2DclassicMulti}(c)), the conditions $h(\vec k)_{11} = h(\vec k)_{22} = h(\vec k)_{33}$ and $|h(\vec k)_{12}| = |h(\vec k)_{23}| = |h(\vec k)_{13}|$ are satisfied only at $\vec k = \vec 0$, which, luckily, is always a minimum \cite{kagomeJ1AF,kagomeJ1J2}. These models have extensive minimum distributions in some regimes, which are the topic of the following section.

Regarding robust non-coplanar spin orders, namely non-coplanar ground states that are not degenerate with any coplanar ones, we must search multi-spin models that do not fulfill the requirement in Theorem~\ref{theorem:mspiral}.

\emph{Extensive spectral minima.}---Spectra with infinitely many minima are intimately connected with disorder, localization, frustration, etc., and in these contexts knowing the entire GSMs is a primary goal. In retrospect, we have determined GSMs by putting together the legitimacy conditions, which enforce unit spin length, and the energy minimization conditions, which rule out certain $\vec S_{\vec k}$ or $\vec S_{\vec k \alpha}$. The legitimacy conditions are decoupled in $\vec n$ space and coupled in $\vec k$ space, and previously we have sacrificed the decoupled form for the simplicity of the energy minimization conditions in $\vec k$ space. In the case of extensive spectral minima, the latter conditions often remain simple in real space. It is hence no good idea to adhere to $\vec S_{\vec k}$ or $\vec S_{\vec k \alpha}$. (Nonetheless, passing to $\vec k$ space is still an important intermediate step, because otherwise we cannot even determine the energy minimization conditions.)

\begin{figure}[!t]
\centering
\includegraphics[width=2.8in]{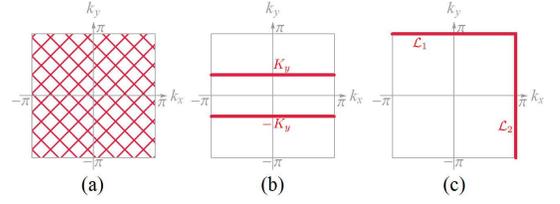}
\caption{Spectral minima covering (a) the entire BZ, (b) two horizontal lines in the BZ, and (c) the BZ boundaries.}
\label{fig:2Dinfinite}
\end{figure}
To illustrate the point, consider the distributions in Fig.~\ref{fig:2Dinfinite} for $N_x \times N_y$ lattices with one spin per unit cell. For Fig.~\ref{fig:2Dinfinite}(a), in terms of $\vec S_{\vec n}$, the legitimacy conditions are simply Eq.~(\ref{1_domain}), and energy minimization says nothing, so the GSM is given by Eq.~(\ref{1_domain}). For Fig.~\ref{fig:2Dinfinite}(b), defining $\vec S_{n_x k_y} \equiv \frac{1}{\sqrt{N_x}} \sum_{k_x} \vec S_{\vec k} e^{i k_x n_x}$, we derive the conditions $\sum_{k_y} \vec S_{n_x k_y}^\dag \vec S_{n_x k_y + p_y} = N_y \delta_{0 p_y}$, $\vec S_{n_x, -k_y} = \vec S_{n_x, k_y}^*$ for legitimacy, and $\vec S_{n_x k_y} = \vec 0, ~ \forall k_y \neq \pm K_y$ for energy minimization, whence the GSM can be trivially computed. Now consider the less trivial situation Fig.~\ref{fig:2Dinfinite}(c), where the interaction is not purely local in any direction, even when restricted to the GSM. This distribution arises in, for instance, the spectrum $\epsilon_{\vec k} = \cos k_x + \cos k_y + \frac 12 [\cos(k_x - k_y) + \cos(k_x + k_y)]$. Energy minimization can be incorporated into the variables $\vec S_{n_x, k_y= \pi} \equiv \frac{1}{\sqrt {N_x}} \sum_{k_x} \vec S_{k_x, k_y=\pi} e^{i k_x n_x}, ~ \vec S_{k_x = \pi, n_y} \equiv \frac{1}{\sqrt{N_y}} \sum_{k_y} \vec S_{k_x=\pi, k_y} e^{i k_y n_y}$, which are independent except that the two $\vec S_{\vec k=(\pi,\pi)}$ must be identified. Plugging them into Eq.~(\ref{1_image}), we find that the GSM is given by \cite{supp},
\begin{eqnarray}
&&\mu_{n_x}, \nu_{n_y}, \vec S_{\vec k = (\pi,\pi)} \in \mathbb R^3, \nonumber \\
&&N_x|\mu_{n_x}|^2 + N_y |\nu_{n_y}|^2 = N + |\vec S_{\vec k= (\pi,\pi)}|^2, \nonumber \\
&&|\mu_{n_x}| \mbox{ independent of } n_x, ~~ |\nu_{n_y}| \mbox{ independent of } n_y, \nonumber \\
&&\frac{1}{\sqrt N} \sum_{\vec n} \mu_{n_x} \cdot \nu_{n_y} e^{-i \vec p \cdot \vec n} = 0, ~ \forall p_x, p_y \neq 0, \nonumber \\
&&\frac{1}{\sqrt{N_x}} \sum_{n_x} \mu_{n_x} = \frac{1}{\sqrt {N_y}} \sum_{n_y} \nu_{n_y} = \vec S_{\vec k=(\pi,\pi)}, \label{e_1cGSM}
\end{eqnarray}
where $\mu_{n_x} \equiv (-)^{n_x} \vec S_{n_x, k_y = \pi}$ and $\nu_{n_y} \equiv (-)^{n_y} \vec S_{k_x = \pi, n_y}$, which is not just the sum of the two submanifolds developed from the edges $\mathcal L_1, \mathcal L_2$ individually. (We will obtain a simpler subset of the GSM if we extend the fourth line of Eq.~(\ref{e_1cGSM}) to all $\vec p \neq \vec 0$. In that case $\mu_{n_x} \cdot \nu_{n_y}$ is independent of $n_x, n_y$, and we can easily enumerate all the possible orientations of $\mu_{n_x}$ and $\nu_{n_y}$.)

In the multi-spin case, to preserve the simplicity of the energy minimization conditions, we undo the Fourier transform but not the unitary transformation by $\xi_{\vec k \alpha}$, thereby bringing in convolution structures, which can be handled numerically. Consider the minimum distributions in Fig.~\ref{fig:2Dinfinite} for non-degenerate lowest band $\alpha = 1$, and assume the requirement in Theorem~\ref{theorem:mspiral} is met. The GSM for Fig.~\ref{fig:2Dinfinite}(a) is given by \cite{supp}
\begin{eqnarray}
(\vec S \ast \xi)_{\vec n, \alpha=1, a}^\dag (\vec S \ast \xi)_{\vec n, \alpha=1, a} = 1, ~~
\vec S_{\vec n, \alpha = 1} \in \mathbb R^3, \label{e_maGSM}
\end{eqnarray}
where
\begin{eqnarray}
(\vec S \ast \xi)_{\vec n \alpha a} &\equiv& \frac{1}{\sqrt N} \sum_{\vec n'} \vec S_{\vec n' \alpha} \xi_{\vec n - \vec n' \alpha a}, \label{e_mSxin}\\
\xi_{\vec n \alpha} &\equiv& \frac{1}{\sqrt{N}} \sum_{\vec k} \xi_{\vec k \alpha} e^{i \vec k \cdot \vec n}, \label{e_mxin} \\
\vec S_{\vec n \alpha} &\equiv& \frac{1}{\sqrt{N}} \sum_{\vec k} \vec S_{\vec k \alpha} e^{i \vec k \cdot \vec n}. \label{e_mSn}
\end{eqnarray}
Note that we must take $\vec S_{\vec n, \alpha = 1}$, not $(\vec S \ast \xi)_{\vec n, \alpha=1, a}$, as independent variables in Eq.~(\ref{e_maGSM}). Similarly, the GSM for Fig.~\ref{fig:2Dinfinite}(b) is given by \cite{supp}
\begin{eqnarray}
&&\sum_{k_y} (\vec S \ast \xi)_{n_x k_y, \alpha=1, a}^\dag (\vec S \ast \xi)_{n_x k_y+p_y, \alpha=1, a} = N_y \delta_{0 p_y}, \nonumber\\
&&\vec S_{n_x,-K_y, \alpha = 1} = \vec S_{n_x, K_y, \alpha=1}^*, \label{e_mbGSM}
\end{eqnarray}
where
\begin{eqnarray}
(\vec S \ast \xi)_{n_x k_y \alpha a} &\equiv& \frac{1}{\sqrt {N_x}} \sum_{n_x'} \vec S_{n_x' k_y \alpha} \xi_{n_x - n_x' k_y \alpha a}, \label{e_mSxinxky}\\
\xi_{n_x k_y \alpha} &\equiv& \frac{1}{\sqrt{N_x}} \sum_{k_x} \xi_{\vec k \alpha} e^{i k_x n_x}, \label{e_mxinxky} \\
\vec S_{n_x k_y \alpha} &\equiv& \frac{1}{\sqrt{N_x}} \sum_{k_x} \vec S_{\vec k \alpha} e^{i k_x n_x}. \label{e_mSnxky}
\end{eqnarray}
As for Fig.~\ref{fig:2Dinfinite}(c), we get \cite{supp}
\begin{align}
&\vec S_{n_x, k_y=\pi, \alpha =1}, ~ \vec S_{k_x=\pi, n_y, \alpha =1}, ~ \vec S_{\vec k = (\pi,\pi), \alpha=1} \in \mathbb R^3, \nonumber \\
&N_x|\mu_{n_x a}|^2 + N_y |\nu_{n_y a}|^2 = N + |\vec S_{\vec k = (\pi,\pi),\alpha=1} \xi_{\vec k = (\pi,\pi),\alpha=1,a}|^2, \nonumber \\
&|\mu_{n_x a}| \mbox{ independent of } n_x, ~~ |\nu_{n_y a}| \mbox{ independent of } n_y, \nonumber \\
&\frac{1}{\sqrt N} \sum_{\vec n} \mu_{n_x a} \cdot \nu_{n_y a} e^{-i \vec p \cdot \vec n} =0, ~ \forall p_x, p_y \neq 0, \nonumber \\
&\frac{1}{\sqrt{N_x}} \sum_{n_x} (-)^{n_x} \vec S_{n_x, k_y=\pi, \alpha =1} = \vec S_{\vec k = (\pi,\pi),\alpha=1}, \nonumber \\
&\frac{1}{\sqrt{N_y}} \sum_{n_y} (-)^{n_y} \vec S_{k_x=\pi, n_y, \alpha =1} = \vec S_{\vec k = (\pi,\pi),\alpha=1}, \label{e_mcGSM}
\end{align}
where
\begin{eqnarray}
\mu_{n_x a} &\equiv& (-)^{n_x} (\vec S \ast \xi)_{n_x, k_y=\pi, \alpha = 1, a}, \label{e_mcmunxa} \\
\nu_{n_y a} &\equiv& (-)^{n_y} (\vec S \ast \xi)_{k_x = \pi, n_y, \alpha = 1, a}, \label{e_mcnunya} \\
(\vec S \ast \xi)_{k_x n_y \alpha a} &\equiv& \frac{1}{\sqrt {N_y}} \sum_{n_y'} \vec S_{k_x n_y' \alpha} \xi_{k_x n_y - n_y' \alpha a}, \label{e_mSxikxny}\\
\xi_{k_x n_y \alpha} &\equiv& \frac{1}{\sqrt{N_y}} \sum_{k_y} \xi_{\vec k \alpha} e^{i k_y n_y}, \label{e_mxikxny} \\
\vec S_{k_x n_y \alpha} &\equiv& \frac{1}{\sqrt{N_y}} \sum_{k_y} \vec S_{\vec k \alpha} e^{i k_y n_y}. \label{e_mSkxny}
\end{eqnarray}

\begin{figure}[!t]
\centering
\includegraphics[width=2in]{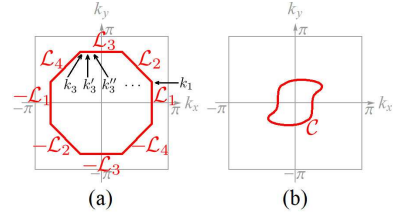}
\caption{Examples of more complicated spectral minimum trajectories.}
\label{fig:2DinfiniteIrregular}
\end{figure}
For even more complicated minimum distributions, the GSMs may be inferred from the sub-GSMs obtained numerically from finite meshes on the minimum trajectories or zones. For this purpose, prior analytical analysis is often labor-saving. Consider for instance the spectral minimum distribution in Fig.~\ref{fig:2DinfiniteIrregular}(a) for one-spin case, and suppose that in a ground state we have $\vec S_{\pm \vec k_1} \neq \vec 0$. Then in light of the fourth entry in Table~\ref{table:2Dfinite}, which states that two generic pairs of minima interfere each other \cite{supp}, we can convince ourselves that $\vec S_{\pm \vec k_3} = \vec 0$, and by induction, $\vec S_{\pm \vec k_3'}, \vec S_{\pm \vec k_3''}, \ldots$ and eventually all $\vec S_{\vec k}$ on $\pm \mathcal L_3$ must vanish. It follows that if any point, not necessarily $\vec k_1$, on edge $\mathcal L_1$ contributes to a ground state, then the entire $\mathcal L_3$ does not. Analysis like this definitely simplifies the numerics. In the limit of infinite number of edges, that is when the trajectory becomes such a random curve $\mathcal C$ as the one in Fig.~\ref{fig:2DinfiniteIrregular}(b), no two pairs are compatible, and the GSM is simply
\begin{equation}
\vec S_{\vec n} = \rho \begin{pmatrix} \cos(\vec k_j \cdot \vec n) \\ -\sin(\vec k_j \cdot \vec n) \\ 0 \end{pmatrix}, ~~ \rho \in SO_3, ~~\vec k_j \in \mathcal C.
\end{equation}

We remark that the seemingly oversimplified situations exemplified by Fig.~\ref{fig:2Dinfinite} and Fig.~\ref{fig:2DinfiniteIrregular}(b) actually cover a few familiar models. Fig.~\ref{fig:2Dinfinite}(a) is the case of the checkerboard model (see Fig.~\ref{fig:2DclassicMulti}; same below) with $J_1 \neq 0, J_2 = |J_1|$ \cite{checkerboardJ1J2AF}, and the Kagome model with $J_1 > 0$ \cite{kagomeJ1AF,kagomeJ1J2}. Fig.~\ref{fig:2Dinfinite}(c) occurs in the square lattice $J_1- J_2$ model with $J_1 > 0, J_2 = \frac 12 J_1$ (see e.g. \cite{squareJ1J2AF,squareJ1J2F}), and the checkerboard model with $J_2 > |J_1|$ \cite{checkerboardJ1J2AF}. The square lattice model with $J_1 < 0, J_2 = \frac 12 |J_1|$ (see e.g. \cite{squareJ1J2F}), and the checkerboard model with $J_1 = 0, J_2 <0$, have a similar distribution in which the minima lie on $k_x = 0$ and $k_y = 0$. Last but not least, Fig.~\ref{fig:2DinfiniteIrregular}(b) reflects the minimum distributions of the honeycomb model with $J_1 \neq 0$ and $0.17 |J_1| < J_2 < 0.48|J_1|$, for which there is one loop, or $J_2 > 0.48|J_1|$, for which there are two loops (numbers are approximate) \cite{localFrame,honeycombJ1J2J3}.

\emph{Anisotropic interactions.}---Consider XXZ Hamiltonians of the following form:
\begin{eqnarray}
H = \frac 12 \sum_{\vec n} \sum_{\eta} J_{\vec \eta} \vec S_{\vec n}^\dag \Theta \vec S_{\vec n + \vec \eta}, \label{a_1}
\end{eqnarray}
or
\begin{equation}
H = \frac 12 \sum_{\vec n} \sum_{\vec \eta} \sum_{a,b=1}^m J_{\vec \eta a b} \vec S_{\vec n a}^\dag \Theta \vec S_{\vec n + \vec \eta b}, \label{a_m}
\end{equation}
where $\Theta \equiv {\rm diag}\{ 1, 1, 1+\theta \}$ is independent of $\vec \eta$, $a$, and $b$. Assume the requirement in Theorem~\ref{theorem:mspiral} is met in the isotropic case. Then a positive $\theta$ simply contracts the GSM, whereas a negative $\theta$ may modify the ground states, which we do not plan to discuss. XY model is the XXZ model with $\theta = + \infty$ and can thus be put into our framework.

Z. Xiong thanks Q. Lin for helpful comments. This work is supported by NSF Grant No. DMR-1005541 and NSFC 11074140.


\begin{thebibliography}{99}

\bibitem{HolsteinPrimakoff} T. Holstein and H. Primakoff, Phys. Rev. {\bf 58}, 1098 (1940).

\bibitem{quantumSpinWave} F. Bloch, Z. Phys. {\bf 61}, 206 (1930); P. W. Anderson, Phys. Rev. {\bf 86}, 694 (1952); R. Kubo, {\it ibid.} {\bf 87}, 568 (1952); T. Oguchi, {\it ibid.} {\bf 117}, 117 (1960).

\bibitem{localFrame} E. Rastelli, A. Tassi, and L. Reatto, Physica (Utrecht) {\bf 97B}, 1 (1979).

\bibitem{QAHEmechanism} J. Ye {\it et al.}, Phys. Rev. Lett. {\bf 83}, 3737 (1999); K. Ohgushi, S. Murakami, and N. Nagaosa, Phys. Rev. B {\bf 62}, R6065 (2000).

\bibitem{pairwiseMinimization} C. A. M. Mulder, H. W. Capel, and J. H. H. Perk, Physica (Utrecht) {\bf 112B}, 147 (1982).

\bibitem{clusterMinimization} D. H. Lyons and T. A. Kaplan, J. Phys. Chem. Solids {\bf 25}, 645 (1964).

\bibitem{LT} J. M. Luttinger and L. Tisza, Phys. Rev. {\bf 70}, 954 (1946); J. M. Luttinger, {\it ibid.} {\bf 81}, 1015 (1951).

\bibitem{GLT} D. H. Lyons and T. A. Kaplan, Phys. Rev. {\bf 120}, 1580 (1960).

\bibitem{reviewLTs} T. A. Kaplan and N. Menyuk, Philos. Mag. {\bf 87}, 3711 (2007).

\bibitem{squareJ1J2F} N. Shannon, B. Schmidt, K. Penc, P. Thalmeier, Eur. Phys. J. B {\bf 38}, 599 (2004).

\bibitem{honeycombJ1J2J3} J. B. Fouet, P. Sindzingre, and C. Lhuillier, Eur. Phys. J. B {\bf 20}, 241 (2001).

\bibitem{supp} See Supplemental Material at [to be specified] for detailed calculations.

\bibitem{squareJ1J2AF} P. Chandra and B. Doucot, Phys. Rev. B {\bf 38}, 9335 (1988).

\bibitem{triangularJ1} B. Bernu, P. Lecheminant, C. Lhuillier, and L. Pierre, Phys. Rev. B, {\bf 50}, 10048 (1994).

\bibitem{checkerboardJ1J2AF} B. Canals, Phys. Rev. B {\bf 65}, 184408 (2002).

\bibitem{kagomeJ1AF} V. Elser, Phys. Rev. Lett. {\bf 62}, 2405 (1989).

\bibitem{kagomeJ1J2} J.-C. Domenge, P. Sindzingre, C. Lhuiller, and L. Pierre, Phys. Rev. B {\bf 72}, 024433 (2005).

\end{thebibliography}
\end{document}


\title{Supplemental Material for ``General method for finding ground state manifold of classical Heisenberg model"}

\author{Zhaoxi Xiong$^{1,}$}
\email{xiong@mit.edu}
\author{Xiao-Gang Wen$^{1,2}$}
\affiliation{$^1$Department of Physics, Massachusetts Institute of Technology, Cambridge, Massachusetts 02139, USA\\
$^2$Perimeter Institute for Theoretical Physics, Waterloo, Ontario, N2L 2Y5 Canada}
\date{\today}

\maketitle

\section{SELECTED EXAMPLES FROM TABLE~I}

\subsection{Distribution with one generic pair of minima}

Let $\pm \vec k_1$ be the spectral minima, and set $\vec S_{\pm \vec k_1} = x \pm i u$, where $x$ and $u$ are real, 3-component column vectors. The second condition in Eq.~(6) is automatically satisfied. The third condition stipulates
\begin{equation}
|x|^2 + |u|^2 = \frac N2 \label{1generic_p=0}
\end{equation}
for $\vec p = \vec 0$, and
\begin{equation}
|x|^2 - |u|^2 = 0, ~~ x \cdot u = 0 \label{1generic_p=2k1}
\end{equation}
for $\vec p = 2 \vec k_1$ or $- 2 \vec k_1$. The solutions are apparently
\begin{equation}
\vec S_{\pm \vec k_1} = \rho \frac{\sqrt N}{2} \begin{pmatrix} 1 \\ \pm i \\ 0 \end{pmatrix}, \label{1generic_Sk}
\end{equation}
or in real space,
\begin{equation}
\vec S_{\vec n} = \rho \begin{pmatrix} \cos(\vec k_1 \cdot \vec n) \\ - \sin (\vec k_1 \cdot \vec n) \\ 0 \end{pmatrix}, \label{1generic_Sn}
\end{equation}
where $\rho \in SO_3$ is an arbitrary rotation.

\subsection{Distribution with two generic pairs of minima}

Let $\pm \vec k_1, ~ \pm \vec k_2$ be the spectral minima, and set $\vec S_{\pm \vec k_1} = x_1 \pm i u_1, ~ \vec S_{\pm \vec k_2} = x_2 \pm i u_2$, where $x_1, u_1, x_2, u_2$ are real, 3-component vectors. Plugging these into the third condition in Eq.~(6), we get
\begin{align}
&|x_1|^2 + |u_1|^2 + |x_2|^2 + |u_2|^2 = \frac N2 & (\vec p = \vec 0), \label{2generic_p=0}\\
&|x_1|^2 - |u_1|^2 = 0, ~ x_1 \cdot u_1 = 0 & (\vec p = \pm 2 \vec k_1), \label{2generic_p=2k1}\\
&|x_2|^2 - |u_2|^2 = 0, ~ x_2 \cdot u_2 = 0 & (\vec p = \pm 2 \vec k_2), \label{2generic_p=2k2}
\end{align}
and
\begin{align}
x_1 \cdot x_2 - u_1 \cdot u_2 = 0, ~ x_2 \cdot u_1 + x_1 \cdot u_2 = 0  & \nonumber \\
(\vec p = \pm (\vec k_1 &+ \vec k_2)), \label{2generic_p=k1+k2}\\
x_1 \cdot x_2 + u_1 \cdot u_2 = 0, ~ x_2 \cdot u_1 - x_1 \cdot u_2 = 0  & \nonumber \\
(\vec p = \pm (\vec k_1 &- \vec k_2)). \label{2generic_p=k1-k2}
\end{align}
Eqs.~(\ref{2generic_p=0})-(\ref{2generic_p=k1-k2}) simplify to
\begin{eqnarray}
&& |x_1| = |u_1|, ~ |x_2| = |u_2|, \\
&&|x_1|^2 + |u_1|^2 + |x_2|^2 + |u_2|^2 = \frac N2, \\
&& x_1, ~ u_1, ~ x_2, ~ u_2 \mbox{ are mutually orthogonal.}
\end{eqnarray}
No four nontrivial, mutually orthogonal vectors can coexist in $\mathbb R^3$. Therefore, either $\vec S_{\pm \vec k_1}$ or $\vec S_{\pm \vec k_2}$ need to vanish. The ground state manifold (GSM) must be
\begin{equation}
\vec S_{\vec n} = \rho \begin{pmatrix} \cos(\vec k_j \cdot \vec n) \\ - \sin (\vec k_j \cdot \vec n) \\ 0 \end{pmatrix}, ~~~ \rho \in SO_3, ~ j = 1,2. \label{2generic_Sn}
\end{equation}

\subsection{Distribution in Fig.~2(i)}

Define $\vec k_1 \equiv (0,0)$ and $\vec k_2 \equiv (\pi, 0)$. The second condition in Eq.~(6) implies $\vec S_{\vec k_1}$ and $\vec S_{\vec k_2}$ are real. The third condition requires
\begin{align}
& |\vec S_{\vec k_1}|^2 + |\vec S_{\vec k_2}|^2 = N & (\vec p = \vec 0), \label{fig2i_p=0}\\
& \vec S_{\vec k_1} \cdot \vec S_{\vec k_2} = 0 & (\vec p = (\pi, 0)). \label{fig2i_p=(pi,0)}
\end{align}
The solutions are
\begin{equation}
\vec S_{\vec k_1} = \rho \sqrt N \begin{pmatrix} \cos \gamma \\ 0 \\ 0 \end{pmatrix}, ~~ \vec S_{\vec k_2} = \rho \sqrt N \begin{pmatrix} 0 \\ \sin \gamma \\ 0 \end{pmatrix}, \label{fig2i_Sk}
\end{equation}
where $\rho \in SO_3$ and $\gamma \in \mathbb R$, so the GSM is
\begin{equation}
\vec S_{\vec n} = \rho \begin{pmatrix} \cos \gamma \\ (-)^{n_x} \sin \gamma \\ 0 \end{pmatrix}. \label{fig2i_Sn}
\end{equation}

\section{EXTENSIVE SPECTRAL MINIMUM DISTRIBUTIONS}

\subsection{Derivation of Eq.~(14)}

Define two families of new variables:
\begin{eqnarray}
\vec S_{n_x k_y} &=& \frac{1}{\sqrt {N_x}} \sum_{k_x} \vec S_{\vec k} e^{i k_x n_x}, \label{eq14_Snxky} \\
\vec S_{k_x n_y} &=& \frac{1}{\sqrt{N_y}} \sum_{k_y} \vec S_{\vec k} e^{i k_y n_y}. \label{eq14_Skxny}
\end{eqnarray}
The energy minimization conditions can be stated as follows: $\vec S_{n_x k_y} = \vec 0, ~ \forall k_y \neq \pi$, and $\vec S_{k_x n_y} = \vec 0, ~ \forall k_x \neq \pi$. Now we want to express the conditions in Eq.~(6) in terms of $\vec S_{n_x,k_y = \pi}$, $\vec S_{k_x=\pi, n_y}$, and $\vec S_{\vec k = (\pi,\pi)}$. The first two conditions go over into
\begin{eqnarray}
&&\vec S_{n_x,k_y = \pi}, ~ \vec S_{k_x=\pi, n_y}, ~ \vec S_{\vec k = (\pi,\pi)} \in \mathbb R^3, \label{eq14_real} \\
&&\frac{1}{\sqrt{N_x}} \sum_{n_x} (-)^{n_x} \vec S_{n_x, k_y=\pi} = \vec S_{\vec k=(\pi,\pi)} \label{eq14_consistencynxky}\\
&&\frac{1}{\sqrt {N_y}} \sum_{n_y} (-)^{n_y} \vec S_{k_x=\pi, n_y} = \vec S_{\vec k=(\pi,\pi)} \label{eq14_consistencykxny}
\end{eqnarray}
The third condition says, for $\vec p = \vec 0$,
\begin{eqnarray}
&&N + |\vec S_{\vec k = (\pi, \pi)}|^2 \nonumber\\
&=& \sum_{k_x} \vec S_{k_x,k_y=\pi}^\dag \vec S_{k_x,k_y=\pi} + \sum_{k_y} \vec S_{k_x=\pi,k_y}^\dag \vec S_{k_x=\pi,k_y} \nonumber\\
&=& \sum_{n_x} \vec S_{n_x,k_y=\pi}^\dag \vec S_{n_x,k_y=\pi} + \sum_{n_y} \vec S_{k_x=\pi,n_y}^\dag \vec S_{k_x=\pi,n_y} \nonumber\\
&=& \sum_{n_x} |\vec S_{n_x,k_y=\pi}|^2 + \sum_{n_y} |\vec S_{k_x=\pi,n_y}|^2. \label{eq14_p=0}
\end{eqnarray}
For $p_x \neq 0, p_y = 0$, only $\mathcal L_1$ contributes (see Fig.~4). Thus,
\begin{eqnarray}
0 &=& \sum_{k_x} \vec S_{k_x,k_y=\pi}^\dag \vec S_{k_x + p_x, k_y=\pi} \nonumber\\
&=& \frac{1}{N_x} \sum_{k_x} \sum_{n_x'} \vec S_{n_x', k_y=\pi} e^{i k_x n_x'} \sum_{n_x} \vec S_{n_x, k_y=\pi} e^{-i (k_x+p_x) n_x} \nonumber\\
&=& \sum_{n_x} \vec S_{n_x,k_y=\pi}^\dag \vec S_{n_x,k_y=\pi} e^{- i p_x n_x}, \label{eq14_p=(px,0)}
\end{eqnarray}
which means that $|\vec S_{n_x,k_y=\pi}|^2$ is independent of $n_x$. Similarly, for $p_x = 0, p_y \neq 0$, we have
\begin{equation}
0 = \sum_{n_y} \vec S_{k_x=\pi,n_y}^\dag \vec S_{k_x=\pi,n_y} e^{- i p_y n_y}, \label{eq14_p=(0,py)}
\end{equation}
which means $|\vec S_{k_x=\pi,n_y}|^2$ is independent of $n_y$. Lastly, a $\vec p$ such that $p_x \neq 0, p_y \neq 0$ couples one point on $\mathcal L_1$ to another point on $\mathcal L_2$. Thus, we have
\begin{eqnarray}
0 &=& \vec S_{k_x=\pi-p_x, k_y = \pi}^\dag \vec S_{k_x=\pi, k_y = \pi+p_y} \nonumber\\
&=& \frac{1}{\sqrt{N}} \sum_{n_x} \vec S_{n_x,k_y=\pi}^\dag e^{i (\pi - p_x) n_x} \sum_{n_y} \vec S_{k_x=\pi,n_y} e^{-i (\pi + p_y) n_y} \nonumber\\
&=& \frac{1}{\sqrt N} \sum_{n_x, n_y} (-)^{n_x+n_y} \vec S_{n_x,k_y=\pi}^\dag \vec S_{k_x=\pi,n_y} e^{-i \vec p \cdot \vec n}, \label{eq14_p=(px,py)}
\end{eqnarray}
In summary, the GSM is given by the following,
\begin{eqnarray}
&&\mu_{n_x}, \nu_{n_y}, \vec S_{\vec k = (\pi,\pi)} \in \mathbb R^3, \nonumber \\
&&N_x|\mu_{n_x}|^2 + N_y |\nu_{n_y}|^2 = N + |\vec S_{\vec k= (\pi,\pi)}|^2, \nonumber \\
&&|\mu_{n_x}| \mbox{ independent of } n_x, ~~ |\nu_{n_y}| \mbox{ independent of } n_y, \nonumber \\
&&\frac{1}{\sqrt{N}} \sum_{\vec n} \mu_{n_x} \cdot \nu_{n_y} e^{-i \vec p \cdot \vec n} = 0, ~ \forall p_x, p_y \neq 0, \nonumber \\
&&\frac{1}{\sqrt{N_x}} \sum_{n_x} \mu_{n_x} = \frac{1}{\sqrt {N_y}} \sum_{n_y} \nu_{n_y} = \vec S_{\vec k=(\pi,\pi)}, \label{eq14_GSM}
\end{eqnarray}
where $\mu_{n_x} \equiv (-)^{n_x} \vec S_{n_x, k_y = \pi}$ and $\nu_{n_y} \equiv (-)^{n_y} \vec S_{k_x = \pi, n_y}$.

\subsection{Derivation of Eq.~(15)}

Define $\vec S_{\vec n \alpha} \equiv \frac{1}{\sqrt N} \sum_{\vec k} \vec S_{\vec k \alpha} e^{i \vec k \cdot \vec n}$ and $\xi_{\vec n \alpha} \equiv \frac{1}{\sqrt N} \sum_{\vec k} \xi_{\vec k \alpha} e^{i \vec k \cdot \vec n}$. The energy minimization conditions are just $\vec S_{\vec n \alpha} = \vec 0, \forall \alpha \neq 1$. As for legitimacy, the first two conditions in Eq.~(13) become
\begin{equation}
\vec S_{\vec n \alpha} \in \mathbb R^3. \label{eq15_real}
\end{equation}
Fourier transforming the third condition, we find
\begin{eqnarray}
N &=& \frac 1N \sum_{\vec p} \sum_{\vec k} \sum_{\alpha',\alpha} \sum_{\vec n', \vec n} \xi_{\vec k \alpha' a}^* \vec S_{\vec n' \alpha'}^\dag \vec S_{\vec n \alpha} \xi_{\vec k + \vec p \alpha a} \nonumber \\
& \times & e^{i \vec k \cdot \vec n'} e^{-i (\vec k + \vec p) \cdot \vec n} e^{i \vec p \cdot \vec n''} \nonumber \\
&=& \frac{1}{\sqrt N} \sum_{\vec k} \sum_{\alpha', \alpha} \sum_{\vec n', \vec n} \xi_{\vec k \alpha' a}^* \vec S_{\vec n' \alpha'}^\dag \vec S_{\vec n \alpha} \xi_{\vec n'' - \vec n \alpha a} e^{-i \vec k \cdot (\vec n'' - \vec n')} \nonumber \\
&=& \left( \sum_{\vec n',\alpha'} \vec S_{\vec n' \alpha'}^\dag \xi_{\vec n''-\vec n' \alpha' a}^* \right) \left( \sum_{\vec n, \alpha} \vec S_{\vec n \alpha} \xi_{\vec n'' - \vec n \alpha a} \right),
\label{eq15_p}
\end{eqnarray}
for all $\vec n''$. Combining Eqs.~(\ref{eq15_real})(\ref{eq15_p}) with the energy minimization conditions, we get
\begin{eqnarray}
(\vec S \ast \xi)_{\vec n, \alpha=1, a}^\dag (\vec S \ast \xi)_{\vec n, \alpha=1, a} = 1, ~~
\vec S_{\vec n, \alpha = 1} \in \mathbb R^3,
\end{eqnarray}
which will determine the GSM, where $(\vec S \ast \xi)_{\vec n \alpha a} \equiv \frac{1}{\sqrt N} \sum_{\vec n'} \vec S_{\vec n' \alpha} \xi_{\vec n - \vec n' \alpha a}$.

\subsection{Derivation of Eq.~(16)}

Define $\vec S_{n_x k_y \alpha} \equiv \frac{1}{\sqrt{N_x}} \sum_{k_x} \vec S_{\vec k \alpha} e^{i k_x n_x}$. The energy minimization conditions are $\vec S_{n_x k_y \alpha} = \vec 0, ~ \forall k_y \neq \pm K_y$ or $\alpha \neq 1$. The first two conditions in Eq.~(13) go over into
\begin{equation}
\vec S_{n_x, -k_y, \alpha} = \vec S_{n_x,k_y,\alpha}^*. \label{eq16_conjugate}
\end{equation}
Fourier transforming the third condition in Eq.~(13) in the $x$ direction, we get
\begin{eqnarray}
N \delta_{0 p_y} &=& \frac{1}{N_x} \sum_{p_x} \sum_{\vec k} \sum_{\alpha',\alpha} \sum_{n_x', n_x} \xi_{\vec k \alpha' a}^* \vec S_{n_x' k_y \alpha'}^\dag \vec S_{n_x k_y+p_y \alpha}  \nonumber \\
& \times & \xi_{\vec k + \vec p \alpha a} e^{i k_x n_x'} e^{-i (k_x + p_x) n_x} e^{i p_x n_x''} \nonumber \\
&=& \frac{1}{\sqrt {N_x}} \sum_{\vec k} \sum_{\alpha', \alpha} \sum_{n_x', n_x} \xi_{\vec k \alpha' a}^* \vec S_{n_x' k_y \alpha'}^\dag \vec S_{n_x k_y+p_y \alpha} \nonumber \\
&\times& \xi_{n_x''-n_x k_y+p_y \alpha a} e^{-i k_x (n_x'' - n_x')} \nonumber \\
&=& \sum_{k_y} \left( \sum_{n_x',\alpha'} \vec S_{n_x' k_y \alpha'}^\dag \xi_{n_x''-n_x' k_y \alpha' a}^* \right) \nonumber \\
&& \left( \sum_{n_x, \alpha} \vec S_{n_x k_y+p_y \alpha} \xi_{n_x''-n_x k_y + p_y \alpha a} \right), \label{eq16_p}
\end{eqnarray}
where $\xi_{n_x k_y \alpha} \equiv \frac{1}{\sqrt{N_x}} \sum_{k_x} \xi_{\vec k \alpha} e^{i k_x n_x}$. Incorporating the energy minimization conditions into Eqs.~(\ref{eq16_conjugate})(\ref{eq16_p}), we get
\begin{eqnarray}
&&\sum_{k_y} (\vec S \ast \xi)_{n_x k_y, \alpha=1, a}^\dag (\vec S \ast \xi)_{n_x k_y+p_y, \alpha=1, a} = N_y \delta_{0 p_y}, \nonumber\\
&&\vec S_{n_x,-K_y, \alpha = 1} = \vec S_{n_x, K_y, \alpha=1}^*, \label{eq16_GSM}
\end{eqnarray}
where $(\vec S \ast \xi)_{n_x k_y \alpha a} \equiv \frac{1}{\sqrt {N_x}} \sum_{n_x'} \vec S_{n_x' k_y \alpha} \xi_{n_x - n_x' k_y \alpha a}$.

\subsection{Derivation of Eq.~(17)}

Define
\begin{eqnarray}
\vec S_{n_x k_y \alpha} &=& \frac{1}{\sqrt{N_x}} \sum_{k_x} \vec S_{\vec k \alpha} e^{i k_x n_x}, \label{eq17_Snxky} \\
\vec S_{k_x n_y \alpha} &=& \frac{1}{\sqrt{N_y}} \sum_{k_y} \vec S_{\vec k \alpha} e^{i k_y n_y}, \label{eq17_Skxny}
\end{eqnarray}
and
\begin{eqnarray}
\xi_{n_x k_y \alpha} &\equiv& \frac{1}{\sqrt{N_x}} \sum_{k_x} \xi_{\vec k \alpha} e^{i k_x n_x}, \label{eq17_xinxky} \\
\xi_{k_x n_y \alpha} &\equiv& \frac{1}{\sqrt{N_y}} \sum_{k_y} \xi_{\vec k \alpha} e^{i k_y n_y}, \label{eq17_xikxny}
\end{eqnarray}
Energy minimization requires that $\vec S_{n_x k_y \alpha} = \vec 0, ~ \forall k_y \neq \pi$ or $\alpha \neq 1$, and that $\vec S_{k_x n_y \alpha} = \vec 0, ~\forall k_x \neq \pi$ or $\alpha \neq 1$. Consistency implies
\begin{eqnarray}
\frac{1}{\sqrt{N_x}} \sum_{n_x} (-)^{n_x} \vec S_{n_x, k_y=\pi, \alpha =1} &=& \vec S_{\vec k = (\pi,\pi),\alpha=1}, \label{eq17_consistencySnxky} \\
\frac{1}{\sqrt{N_y}} \sum_{n_y} (-)^{n_y} \vec S_{k_x=\pi, n_y, \alpha =1} &=& \vec S_{\vec k = (\pi,\pi),\alpha=1}. \label{eq17_consistencySkxny}
\end{eqnarray}
Again,
\begin{equation}
\vec S_{n_x, k_y=\pi, \alpha =1}, ~ \vec S_{k_x=\pi, n_y, \alpha =1}, ~ \vec S_{\vec k = (\pi,\pi),\alpha=1} \in \mathbb R^3.
\end{equation}
Now we plug $\vec S_{n_x, k_y=\pi, \alpha =1}, \vec S_{k_x=\pi, n_y,\alpha=1}$ into the last condition in Eq.~(13). For $\vec p = \vec 0$, we have
\begin{widetext}
\begin{eqnarray}
&&N + |\vec S_{\vec k = (\pi,\pi), \alpha = 1}|^2 |\xi_{\vec k = (\pi,\pi),\alpha=1,a}|^2 \nonumber\\
&=& \sum_{k_x} |\vec S_{k_x,k_y=\pi,\alpha=1}|^2 |\xi_{k_x,k_y=\pi,\alpha=1,a}|^2
+ \sum_{k_y} |\vec S_{k_x=\pi,k_y,\alpha=1}|^2 |\xi_{k_x=\pi,k_y,\alpha=1,a}|^2 \nonumber\\
&=& \Bigg( \frac{1}{N_x^2} \sum_{k_x} \sum_{n_x',n_x} \sum_{n_x''',n_x''} \vec S_{n_x',k_y=\pi,\alpha=1}^\dag \vec S_{n_x,k_y=\pi,\alpha=1}
\xi_{n_x''',k_y=\pi,\alpha=1,a} \xi_{n_x'',k_y=\pi,\alpha=1,a} e^{ik_x(n_x'-n_x+n_x'''-n_x'')} \Bigg)
+ \ldots \nonumber\\
&=& \Bigg( \frac{1}{N_x} \sum_{n_x',n_x} \sum_{n_x''',n_x''} \vec S_{n_x',k_y=\pi,\alpha=1}^\dag \vec S_{n_x,k_y=\pi,\alpha=1}
\xi_{n_x''',k_y=\pi,\alpha=1,a} \xi_{n_x'',k_y=\pi,\alpha=1,a} \delta_{n_x'-n_x+n_x'''-n_x''} \Bigg)
+ \ldots \nonumber\\
&=& \sum_{n_x} |(\vec S \ast \xi)_{n_x,k_y=\pi,\alpha=1,a}|^2
+ \sum_{n_y} |(\vec S \ast \xi)_{k_x=\pi,n_y,\alpha=1,a}|^2, \label{eq17_p=0}
\end{eqnarray}
where $(\vec S \ast \xi)_{n_x k_y \alpha a} \equiv \frac{1}{\sqrt {N_x}} \sum_{n_x'} \vec S_{n_x' k_y \alpha} \xi_{n_x - n_x' k_y \alpha a}$ and $(\vec S \ast \xi)_{k_x n_y \alpha a} \equiv \frac{1}{\sqrt {N_y}} \sum_{n_y'} \vec S_{k_x n_y' \alpha} \xi_{k_x n_y - n_y' \alpha a}$. For $p_x \neq 0, p_y = 0$, convolutions arise in a similar way:
\begin{eqnarray}
0 &=& \sum_{k_x} \vec S_{k_x,k_y=\pi,\alpha=1}^\dag \vec S_{k_x+p_x,k_y=\pi,\alpha=1}
\xi_{k_x,k_y=\pi,\alpha=1,a}^* \xi_{k_x+p_x,k_y=\pi,\alpha=1,a} \nonumber\\
&=& \frac{1}{N_x^2} \sum_{k_x} \sum_{n_x',n_x} \sum_{n_x''',n_x''} \vec S_{n_x',k_y=\pi,\alpha=1}^\dag \vec S_{n_x,k_y=\pi,\alpha=1}
\xi_{n_x''',k_y=\pi,\alpha=1,a}^* \xi_{n_x'',k_y=\pi,\alpha=1,a}
e^{ik_x(n_x'+n_x''')} e^{-i (k_x+p_x) (n_x + n_x'')} \nonumber\\
&=& \sum_{n_x} |(\vec S \ast \xi)_{n_x,k_y=\pi,\alpha=1,a}|^2 e^{-i p_x n_x}, \label{eq17_p=(px,0)}
\end{eqnarray}
which means $|(\vec S \ast \xi)_{n_x,k_y=\pi,\alpha=1,a}|^2$ is independent of $n_x$. For $p_x=0, p_y \neq 0$, we have
\begin{equation}
0 = \sum_{n_y} |(\vec S \ast \xi)_{k_x=\pi,n_y,\alpha=1,a}|^2 e^{-i p_y n_y}, \label{eq17_p=(0,py)}
\end{equation}
which means $|(\vec S \ast \xi)_{k_x=\pi,n_y,\alpha=1,a}|^2$ is independent of $n_y$. For $p_x \neq 0, p_y \neq 0$, we have
\begin{eqnarray}
0 &=& \vec S_{k_x=\pi-p_x,k_y=\pi,\alpha=1}^\dag \vec S_{k_x=\pi,k_y=\pi+p_y,\alpha=1}
\xi_{k_x=\pi-p_x,k_y=\pi,\alpha=1,a}^* \xi_{k_x=\pi,k_y=\pi+p_y,\alpha=1,a} \nonumber\\
&=& \frac 1N \sum_{n_x,n_y} \sum_{n_x',n_y'} \vec S_{n_x,k_y=\pi,\alpha=1}^\dag \vec S_{k_x=\pi,n_y,\alpha=1}
\xi_{n_x',k_y=\pi,\alpha=1,a}^* \xi_{k_x=\pi,n_y',\alpha=1,a}
e^{i(\pi-p_x)(n_x+n_x')} e^{-i(\pi+p_y)(n_y+n_y')} \nonumber\\
&=& \frac{1}{\sqrt N} \sum_{\vec n} (\vec S \ast \xi)_{n_x,k_y=\pi,\alpha=1,a}^\dag (\vec S \ast \xi)_{k_x=\pi,n_y,\alpha=1,a}
(-)^{n_x+n_y} e^{- i \vec p \cdot \vec n}. \label{eq17_p=(px,py)}
\end{eqnarray}
\end{widetext}
Summarizing, the GSM is determined by the following:
\begin{align}
&\vec S_{n_x, k_y=\pi, \alpha =1}, ~ \vec S_{k_x=\pi, n_y, \alpha =1}, ~ \vec S_{\vec k = (\pi,\pi), \alpha=1} \in \mathbb R^3, \nonumber \\
&N_x|\mu_{n_x a}|^2 + N_y |\nu_{n_y a}|^2 = N + |\vec S_{\vec k = (\pi,\pi),\alpha=1} \xi_{\vec k = (\pi,\pi),\alpha=1,a}|^2, \nonumber \\
&|\mu_{n_x a}| \mbox{ independent of } n_x, ~~ |\nu_{n_y a}| \mbox{ independent of } n_y, \nonumber \\
&\frac{1}{\sqrt N} \sum_{\vec n} \mu_{n_x a} \cdot \nu_{n_y a} e^{-i \vec p \cdot \vec n} =0, ~ \forall p_x, p_y \neq 0, \nonumber \\
&\frac{1}{\sqrt{N_x}} \sum_{n_x} (-)^{n_x} \vec S_{n_x, k_y=\pi, \alpha =1} = \vec S_{\vec k = (\pi,\pi),\alpha=1}, \nonumber \\
&\frac{1}{\sqrt{N_y}} \sum_{n_y} (-)^{n_y} \vec S_{k_x=\pi, n_y, \alpha =1} = \vec S_{\vec k = (\pi,\pi),\alpha=1}, \label{eq17_GSM}
\end{align}
where $\mu_{n_x a} \equiv (-)^{n_x} (\vec S \ast \xi)_{n_x, k_y=\pi, \alpha = 1, a}$ and $\nu_{n_y a} \equiv (-)^{n_y} (\vec S \ast \xi)_{k_x = \pi, n_y, \alpha = 1, a}$.

%